\documentclass[11pt]{article}
\def\withcolors{0}
\def\withnotes{0}

\usepackage{ccanonne}
\usepackage{soul}
\usepackage[normalem]{ulem}
\usepackage{authblk}

\usepackage[style=alphabetic,maxbibnames=99,maxalphanames=99]{biblatex}
\title{Optimal Closeness Testing of Discrete Distributions Made \sout{Complex} Simple}

\author[$\dagger$]{Cl\'{e}ment L. Canonne}
\author[$\dagger$]{Yucheng Sun}
\affil[$\dagger$]{University of Sydney}

\begin{document}
\maketitle
\begin{abstract}
    In this note, we revisit the recent work of Diakonikolas, Gouleakis, Kane, Peebles, and Price~\cite{DiakonikolasGKP21}, and provide an alternative proof of their main result. Our argument does not rely on any specific property of Poisson random variables (such as stability and divisibility) nor on any ``clever trick,'' but instead on an identity relating the expectation of the absolute value of any random variable to the integral of its characteristic function:
    \[
        \bEE{|X|} = \frac{2}{\pi}\int_0^\infty \frac{1-\Re(\bEE{e^{i tX}})}{t^2}\, dt
    \]
    Our argument, while not devoid of technical aspects, is arguably conceptually simpler and more general; and we hope this technique can find additional applications in distribution testing.
\end{abstract}
In the \emph{closeness testing} problem, one is given i.i.d.\ samples from two unknown probability distributions $\p,\q$ over a known discrete domain of size $\ab$, without loss of generality $[\ab] \eqdef \{1,2,\dots, \ab\}$; along with distance and error parameters $\dst\in(0,1]$ and $\errprob\in(0,1]$. The goal is to find the minimum number of samples sufficient to distinguish between the two cases (i)~$\p=\q$ and (ii)~$\totalvardist{\p}{\q}>\dst$ and be correct in both cases with probability at least $1-\errprob$ (for all possible inputs $\p,\q$), where
\begin{equation}
    \totalvardist{\p}{\q} = \sup_{S\subseteq [\ab]} \Paren{\p(S)-\q(S)} = \frac{1}{2} \sum_{i=1}^\ab \abs{\p_i - \q_i} \in[0,1]
\end{equation}
denotes the total variation distance (statistical distance). This minimum number of samples, $\ns(\ab,\dst,\errprob)$, is the \emph{sample complexity} of closeness testing; and the optimal dependence on all parameters (including $\errprob$), up to constant factors, was recently obtained by Diakonikolas, Gouleakis, Kane, Peebles, and Price~\cite{DiakonikolasGKP21} (previous work only focused on, and obtained, the right dependence on $\ab,\dst$~\cite{ChanDVV14}).\footnote{For more on closeness (and distribution) testing, the reader is referred to, e.g.,~\cite{Canonne:20} and~\cite[Chapter~11]{Goldreich:17}.}

\begin{theorem}[\cite{DiakonikolasGKP21}]
    \label{theo:optimal:sample:complexity}
The sample complexity of closeness testing is
\begin{equation}
    \ns(\ab,\dst,\errprob) = \bigTheta{ \max\Paren{ \frac{\log(1/\errprob)}{\dst^2}, \frac{\ab^{2/3}\log^{1/3}(1/\errprob)}{\dst^{4/3}},
    \frac{\ab^{1/2}\log^{1/2}(1/\errprob)}{\dst^{2}}}
    }\,.
\end{equation}
\end{theorem}
We here focus exclusively on the upper bound, that is, on the testing algorithm achieving this sample complexity. Suppose we take two sets of $\ns$ i.i.d.\ samples from both $\p$ and $\q$, and for each of those four sets compute the number of occurrences of each of the $\ab$ domain elements among the correspond $\ns$ samples.
 We then consider the (renormalized) statistic introduced in~\cite{DiakonikolasGKP21}:
\begin{equation}
    \label{eq:statistic:z}
    Z \eqdef \frac{1}{\ns}\sum_{i=1}^\ab \mleft( |X_i-Y_i| + |X'_i-Y'_i| - |X_i-X'_i| - |Y_i-Y'_i| \mright)
\end{equation}
where, for each fixed $i$, $X_i,X'_i\sim\binomial{\ns}{\p_i}$ and $Y_i,Y'_i\sim\binomial{\ns}{\q_i}$ are independent (but the $X_i$'s are not independent across different $i$'s). 

In~\cite{DiakonikolasGKP21}, it was shown that the expectation of $Z$ in the cases $\p=\q$ and $\totalvardist{\p}{\q}>\dst$ differed by a noticeable quantity; a comparatively easy argument then allowed them to prove that $Z$ was with high probability close to its expectation; and suitably thresholding this statistic $Z$ to distinguish between the two cases led to the optimal sample complexity.

However, the key part of their argument, which led to establishing this expectation gap between the two cases, was quite unwieldy, and relied on (1)~considering $X_i,X'_i,Y_i,Y'_i$ distributed as Poisson random variables instead of Binomials (i.e., $X_i \sim \poisson{\ns\p_i}$ vs. $X_i\sim \binomial{\ns}{\p_i}$), along with (2)~a clever and intricate use of specific properties of Poisson random variables, such as stability and divisibility. This was then combined with an additional argument establishing that assuming $X_i,X'_i,Y_i,Y'_i$ were Poisson instead of Binomial could be done, in this specific case, without affecting the expectation gap itself (that is, that the gap in expectation between the (analyzed) Poisson and the (true) Binomial cases was smaller that the gap in expectation shown between the $\p=\q$ and $\totalvardist{\p}{\q}>\dst$ cases assuming all random variables were Poisson).

\paragraph{Our contribution.} The goal of this note is to provide an alternative, direct proof of the expectation gap, directly in the usual multinomial setting described above where the random variables are Binomial, and without relying on any particular property of Poisson random variables. To do so, we will rely on the key identity below:
\begin{theorem}[{Zolotarev identity~\cite[Eq. (3.26)]{Pinelis2016}}]
    \label{theo:zolotarev}
For any r.v. $X$, we have
\[
        \bEE{|X|} = \frac{2}{\pi}\int_0^\infty \frac{1-\Re(\bEE{e^{i tX}})}{t^2}\, dt
\]
\end{theorem}
\noindent which does not require any additional condition on the random variables besides their having a well-defined expectation. To the best of our knowledge, this is the first proof of (the upper bound of)~\cref{theo:optimal:sample:complexity} which works directly in the multinomial setting (``standard sampling setting''), does not rely on \emph{ad hoc}, clever (but non-generalizable) tricks, and provides explicit and relatively small (albeit not optimized) constants.

\paragraph{Organization.} We first provide an outline of the main argument and of the use of the Zolotarev identity in~\cref{sec:outline}, before some (short) preliminaries in~\cref{sec:preliminaries}. We then establish the main lemma, the expectation gap, in~\cref{sec:multinomial} (\cref{lemma:main:gap:binomial}), before briefly recalling how this implies (the upper bound of)~\cref{theo:optimal:sample:complexity} in~\cref{sec:getting:expectation:gap:from:there}. We finally show the versatility of our argument by (re)establishing in~\cref{sec:poissonized} the analogue of~\cref{lemma:main:gap:binomial} in the Poissonized setting, i.e., the statement obtained by other means  in~\cite{DiakonikolasGKP21}.\footnote{The proof of~\cref{lemma:main:gap:poisson} in~\cref{sec:poissonized}  is somewhat simpler than that of~\cref{lemma:main:gap:binomial}, due to some tedious technical details in the latter one; but they are conceptually identical.}

\section{Outline and intuition}
    \label{sec:outline}
The statistic $Z$ defined in~\eqref{eq:statistic:z} was specifically designed so that, when $\p=\q$,
\begin{equation}
    \bE{\p\p}{Z} = 0
\end{equation}
so the crux is to prove that 
\begin{equation}
    \label{eq:exp:gap}
  \bE{\p\q}{Z} = \bigOmega{\min\Paren{\dst, \dst^2 \frac{\ns}{\ab}, \dst^2 \sqrt{\frac{\ns}{\ab}}}} 
\end{equation}
whenever $\totalvardist{\p}{\q} \geq \dst$. Recalling the definition of $Z$, by linearity of expectation it will be enough to analyze each of the $\ab$ summands  separately.
Thus, both in the multinomial (\cref{sec:multinomial}) and Poissonized (\cref{sec:poissonized}) sampling models, the key lemma is to show that
\begin{equation}
    \label{eq:key:goal}
    \bEE{|X-Y| + |X'-Y'| - |X-X'| - |Y-Y'|} 
    \gtrsim \min\Paren{ (\mu-\lambda)^2, |\mu-\lambda|, \frac{(\mu-\lambda)^2}{\sqrt{\mu+\lambda}} }\,,
\end{equation}
where $\mu = \bEE{X} = \bEE{X'}$ and $\lambda = \bEE{Y} = \bEE{Y'}$.

Once this inequality is established, the proof for the expectation gap follows from a relatively straightforward distinction of cases, mimicking the last part of the argument of~\cite{DiakonikolasGKP21} (we recall this argument in~\cref{sec:getting:expectation:gap:from:there}). Our key contribution thus lies in establishing~\cref{eq:key:goal}.\medskip

To do so, we invoke Zolotarev's identity to get rid of the absolute values, allowing us to express (exactly) the LHS as the integral of a real-valued, \emph{non-negative} function over $[0,\infty)$:
\[
    \int_0^\infty \frac{dt}{t^2} f(t)
\]
such that $f(t) = \Theta((\lambda-\mu)^2 t^2)$ as $t\to 0$. Since the integrand is non-negative, we can then hope to lower bound the expression by
\[
    \int_0^\tau \frac{dt}{t^2} f(t)
\]
for some suitable $\tau$ chosen so that the asymptotic approximation $f(t) \asymp (\lambda-\mu)^2 t^2$ holds for all $0\leq t\leq \tau$; which would then give us the lower bound
\[
    \int_0^\infty \frac{dt}{t^2} f(t) \geq \int_0^\tau \frac{dt}{t^2} f(t) \asymp \tau (\lambda-\mu)^2
\]
This is exactly what we do, distinguishing three cases for our chose of $\tau$ as a function of the values of $|\lambda-\mu|$ and $\lambda+\mu$. Namely, the three cases correspond to $\tau \asymp 1$, $\tau \asymp 1/|\lambda-\mu|$, and $\tau \asymp 1/\sqrt{\lambda+\mu}$, giving the three terms of~\cref{eq:key:goal}.

\section{Technical preliminaries}
    \label{sec:preliminaries}
The argument will only require minimal knowledge of discrete probability (namely, the expression of the characteristic function (CF) of a Binomial distribution) as well as some (limited) familiarity with complex numbers. 
We also will rely on the following standard fact:
\begin{fact}
    \label{fact:trigo}
For $0\leq t\leq \frac{\pi}{2}$, we have $\frac{2}{\pi} t \leq \sin t \leq t$; and $\cos$ is decreasing on $[0,\pi]$.
\end{fact}

\noindent In terms of notation, we will use $\asymp$ and $\gtrsim$, $\lesssim$ to ignore constants in (in)equalities: i.e., $a_n\gtrsim b_n$ means that there exists an absolute constant $C>0$ such that $a_n \geq b_n$ for all $n$; and $a_n \asymp b_n$ means that both $a_n \gtrsim b_n$ and $a_n \lesssim b_n$ hold.

\begin{remark}
We will for convenience assume throughout that the two unknown distributions have no ``heavy elements,'' i.e., that $\norminf{\p},\norminf{\q} \leq 1/4$. While this may seem restrictive at first sight, this can actually be done without loss of generality by a standard trick, which consists in mapping every element of the domain $[\ab]$ to 4 ``new elements'' $4i-3,4i-2,4i-1,4i$ in a larger domain $[4\ab]$, and (independently) mapping each sample in $[\ab]$ uniformly at random to one of the 4 corresponding elements in $[4\ab]$. This only increases the domain size by a factor $4$,\footnote{One can also use a slightly more involved transformation to avoid paying the resulting constant factor (which is reflected in the sample complexity) by first identifying the (constantly many) elements with probability at least $1/4$ under both $\p$ and $\q$ (by learning the distributions to $\lp[\infty]$ distance $1/8$, e.g., via the Dvoretzky--Kiefer--Wolfowitz inequality); and applying the above mapping to only those few elements.} preserves the total variation distances, and $\ns$ samples from the original distribution $\p$ over $[\ab]$ can be transformed into $\ns$ i.i.d.\ samples from the ``induced'' distribution $\p'$ over $[4\ab]$, which now satisfies $\norminf{\p'} = \norminf{\p}/4 \leq 1/4$.
\end{remark}

\section{Expectation gap in the multinomial setting}
    \label{sec:multinomial}
We start with the result in the ``multinomial'' case, which corresponds to the standard sample setting where exactly $2\ns$ samples are taken from each of $\p$ and $\q$, to obtain $X,X'$ and $Y,Y'$. Recall that we assume throughout $\norminf{\p},\norminf{\q} \leq 1/4$, which will help in some parts of the analysis.
\begin{lemma}
    \label{lemma:main:gap:binomial}
Let $p,q \in [0,1/4]$, and $\ns \geq 16$. Suppose $X,X'\sim\operatorname{Bin}(\ns,p)$ and $Y,Y'\sim\operatorname{Bin}(\ns,q)$ are mutually independent. Then
\begin{align*}
    \bEE{|X-Y| + |X'-Y'| - |X-X'| - |Y-Y'|} 
    &\geq \min\Paren{\frac{1}{8}\ns^2(p-q)^2, \frac{1}{16} \ns|p-q|, \frac{\ns^2(p-q)^2}{40\sqrt{\ns(p+q)}}} \\
    &\asymp \min\Paren{ (\mu-\lambda)^2, |\mu-\lambda|, \frac{(\mu-\lambda)^2}{\sqrt{\mu+\lambda}} }\,,
\end{align*}
where $\mu \eqdef \ns p = \bEE{X}$ and $\lambda\eqdef \ns q = \bEE{Y}$.
\end{lemma}
\begin{proof}
Our first step is to use Zolotarev's identity (\cref{theo:zolotarev}) to rewrite the quantity to bound as an integral involving the characteristic functions of $X,Y$, getting rid of the absolute values:
\begin{align}
     \Delta &\eqdef \bEE{|X-Y| + |X'-Y'| - |X-X'| - |Y-Y'|} \notag\\
    &= \frac{2}{\pi}\Re\int_0^\infty \frac{
    \bEE{e^{i t(X-X')}}
    +\bEE{e^{i t(Y-Y')}}
    -2\bEE{e^{i t(X-Y)}}}{t^2}\, dt \notag\\
    &= \frac{2}{\pi}\Re\int_0^\infty \frac{
    \bEE{e^{i tX}}\bEE{e^{-i tX}}
    +\bEE{e^{i tY}}\bEE{e^{-i tY}}
    -2\bEE{e^{i tX}}\bEE{e^{-i tY}}}{t^2}\, dt \notag\\
    &= \frac{2}{\pi}\int_0^\infty \frac{dt}{t^2}\Paren{
    \bEE{e^{i tX}}\bEE{e^{-i tX}}
    +\bEE{e^{i tY}}\bEE{e^{-i tY}}
    -2\Re(\bEE{e^{i tX}}\bEE{e^{-i tY}})} \label{eq:cf}
\end{align}
In particular, letting $u(t) \eqdef \bEE{e^{i tX}}\in\C$, $v(t) \eqdef \bEE{e^{i tY}}\in\C$, we have
\begin{align}
    \Delta 
    &= \frac{2}{\pi}\int_0^\infty \frac{dt}{t^2}\Paren{
    \abs{u(t)}^2+\abs{v(t)}^2
    -2\Re(u(t) \overline{v(t)})}
    = \frac{2}{\pi}\int_0^\infty \frac{dt}{t^2}\underbrace{\abs{u(t)-v(t)}^2}_{\geq 0} \label{eq:integrand:nonneg}
\end{align}
and so $\Delta \geq 0$. Importantly, the fact that the integrand is non-negative means we can choose to restrict the integral to any subset $S\subseteq [0,\infty)$, and still derive a lower bound on $\Delta$.

Recalling the characteristic function of a Binomial distribution, we have $u(t) = (1-p+p e^{i t})^\ns$. 
Writing further $1-p+p e^{i t} = r(t) e^{i \theta(t) }$ and $1-q+q e^{i t} = s(t) e^{i \eta(t) }$
with
\begin{align}
    \label{eq:binomialcase:rhosigma}
    \begin{aligned}
    r(t)  \eqdef \sqrt{(1-p(1-\cos t))^2 + p^2 \sin^2 t} \\
    s(t)  \eqdef \sqrt{(1-q(1-\cos t))^2 + q^2 \sin^2 t}
    \end{aligned}
\end{align}
we get
\begin{align*}
    \Delta 
    &= \frac{2}{\pi}\int_0^\infty \frac{dt}{t^2} \Paren{r^{2\ns}(t)+s^{2\ns}(t)
    -2r(t)^{\ns}s(t)^{\ns}\cos(2\ns\theta(t)-2\ns\eta(t))} \\
    &\geq \frac{4}{\pi}\int_0^\infty \frac{dt}{t^2}  r(t)^{\ns}s(t)^{\ns}\Paren{1-\cos(2\ns(\theta(t)-\eta(t)))}\tag{AM-GM}
\end{align*}

To lower bound this expression when $t$ is small, we need to bound $2\ns(\theta(t)-\eta(t))$ when $0 \leq t \leq \pi$. We do so in the next (slightly technical) claim, whose proof is deferred to the end of the section:
\begin{claim}
    \label{claim:technical:inequality}
For $0\leq t\leq \pi$ and $0\leq p,q\leq 1/4$, we have
\begin{equation}
    2\ns\abs{\theta(t)-\eta(t)} 
    > 2\ns\abs{p-q} \sin t \label{eq:cos:lower_bound}
\end{equation}
and
\begin{equation}
    2\ns\abs{\theta(t)-\eta(t)} 
    \leq 14\ns\abs{p-q} t \label{eq:cos:upper_bound}
\end{equation}
\end{claim}
\noindent We will also use~\cref{fact:trigo} quite extensively to ``replace'' the $\sin t$ of~\cref{eq:cos:lower_bound} by $\frac{2}{\pi}t$ whenever possible.

Where do we go from there? We have three cases, which will determine where to truncate the integral to derive the lower bound on $\Delta$:
\begin{description}
    \item[Case 1:] $\ns(p+q) \leq 1/2$. Then, from~\cref{eq:binomialcase:rhosigma},
    \[
        r(t),\ s(t) \geq 1-\frac{1}{\ns}
    \] 
and so $r(t)^{\ns}s(t)^{\ns} \geq (1-\frac{1}{\ns})^{2\ns} \geq \frac{1}{8}$ (using $\ns \geq 16$). Therefore,
\begin{align*}
    \Delta 
    &\geq \frac{1}{2\pi}\int_0^\infty \frac{dt}{t^2}\cdot  \Paren{1-\cos(2\ns(\theta(t)-\eta(t)))}
\end{align*}
As long as $14 \ns\abs{p-q}\cdot t \leq \pi$ (which is true when $t \leq \frac{\pi}{7}$, given our assumption on $p+q \leq 1/(2\ns)$), by~\cref{eq:cos:lower_bound,fact:trigo} we can write $\cos(2\ns\abs{\theta(t)-\eta(t)}) \leq  \cos\Paren{\frac{4}{\pi} \ns(p-q) t}$. In view of the above, using the fact that the integrand is always non-negative (\cref{eq:integrand:nonneg}), we can further lower bound $\Delta$ as
\begin{align*}
    \Delta 
    &\geq \frac{1}{2\pi}\int_0^{\frac{\pi}{7}} \frac{dt}{t^2}\cdot  \Paren{1-\cos(2\ns(\theta(t)-\eta(t)))}
    \geq \frac{1}{2\pi}\int_0^{\frac{\pi}{7}} \frac{dt}{t^2}\cdot  \Paren{1-\cos\Paren{2 \ns\abs{p-q}\cdot t}} \\
    &= \frac{\ns\abs{p-q}}{\pi} \int_0^{\frac{2 \pi}{7}\ns|p-q|} \frac{du}{u^2}(1-\cos u) 
\end{align*}
Now, since $\frac{2 \pi}{7}\ns|p-q| \leq \frac{\pi}{7}$ and $u\mapsto \frac{du}{u^2}(1-\cos u)$ is decreasing on $[0,2\pi]$, we get
\begin{equation}
    \label{eq:bound:binomial:case1}
    \Delta 
    \geq \frac{\ns\abs{p-q}}{\pi} \cdot \frac{ 2\pi}{7}\ns|p-q| \cdot \frac{1-\cos(\frac{\pi}{7})}{(\frac{\pi}{7})^2} 
    \geq \frac{1}{8}\ns^2(p-q)^2
\end{equation}

    \item[Case 2:] $1/2 < \ns(p+q)$ and $\abs{\ns (p - q)} \geq \sqrt{\ns(p+q)}$. Then we can write, dropping the $\sin^2$ terms in~\cref{eq:binomialcase:rhosigma},
\begin{align*}
    \Delta 
    &\geq \frac{1}{2\pi}\int_0^\infty \frac{dt}{t^2}(1-p(1-\cos t))^\ns(1-q(1-\cos t))^\ns  \Paren{1-\cos(2\ns(\theta(t)-\eta(t)))} \\
    &\geq \frac{1}{2\pi}\int_0^\infty \frac{dt}{t^2}(1-(p+q)(1-\cos t))^\ns  \Paren{1-\cos(2\ns(\theta(t)-\eta(t)))} \\
    &\geq \frac{1}{2\pi}\int_0^\infty \frac{dt}{t^2}(1-\ns(p+q)(1-\cos t))_+  \Paren{1-\cos(2\ns(\theta(t)-\eta(t)))} \tag{$(1-x)^\ns\geq (1-\ns x)_+$}\\
    &\geq \frac{1}{2\pi}\int_0^\infty \frac{dt}{t^2}(1-\ns^2(p-q)^2(1-\cos t))_+  \Paren{1-\cos(2\ns(\theta(t)-\eta(t)))} \tag{Since $p+q\leq \ns(p-q)^2$} \\
    &\geq \frac{1}{2\pi}\int_0^{\frac{\pi}{14\ns|p-q|}} \frac{dt}{t^2}(1-\ns^2(p-q)^2(1-\cos t))  \Paren{1-\cos(2\ns(\theta(t)-\eta(t)))} 
\end{align*}
where for the last second-to-last inequality, we restricted the domain to the interval $[0,\pi /14(\ns|p-q|)]$ and got rid of the $(\cdot)_+$, as on this interval we know that the parenthesis is non-negative (since $1-\cos t \leq t^2/2$). Observe that on this interval we have $14 \ns\abs{(p-q) } t \leq  \pi$, and so again by~\cref{eq:cos:lower_bound,fact:trigo} we get
\begin{align}
    \Delta 
    &\geq \frac{1}{2\pi}\int_0^{\frac{\pi}{14\ns|p-q|}} \frac{dt}{t^2}(1-\ns^2(p-q)^2(1-\cos t))  \Paren{1-\cos(2 \ns{(p-q) t})} \notag\\
    &\geq \frac{1}{2\pi}\int_0^{\frac{\pi}{14\ns|p-q|}} \frac{dt}{t^2}(1-\frac{1}{2}\ns^2(p-q)^2t^2)  \Paren{1-\cos(2 \ns{(p-q) t})} \notag\\
    &= \frac{\ns|p-q|}{2\pi}\int_0^{\frac{\pi}{14}} \frac{du}{u^2}\Paren{1-\frac{u^2}{2}}  \Paren{1-\cos{2u}} \tag{$u \eqdef \ns|p-q|t$}\notag\\
    &\geq \frac{1}{16}\cdot \ns|p-q| \label{eq:bound:binomial:case2}
\end{align}
    \item[Case 3:] $1/2 < \ns(p+q)$ and $\abs{\ns (p - q)} < \sqrt{\ns(p+q)}$. Then we can start as in Case 2, before truncating the integral at $\pi/(14\sqrt{\ns(p+q)})$:
    \begin{align}
    \Delta 
    &\geq \frac{1}{2\pi}\int_0^\infty \frac{dt}{t^2}(1-\ns(p+q)(1-\cos t))_+  \Paren{1-\cos(2\ns(\theta(t)-\eta(t)))} \tag{as in Case 2} \notag\\
    &\geq \frac{1}{2\pi}\int_0^{\frac{\pi}{14\sqrt{\ns(p+q)}}} \frac{dt}{t^2}(1-\ns(p+q)(1-\cos t))_+  \Paren{1-\cos(2\ns(\theta(t)-\eta(t)))} \notag\\
    &\geq \frac{1}{2\pi}\int_0^{\frac{\pi}{14\sqrt{\ns(p+q)}}} \frac{dt}{t^2}(1-\frac{1}{2}\ns(p+q)t^2)  \Paren{1-\cos(2\ns(\theta(t)-\eta(t)))} \tag{$1-\cos t \leq t^2/2$} \notag\\
    &\geq \frac{1}{2\pi}\int_0^{\frac{\pi}{14\sqrt{\ns(p+q)}}} \frac{dt}{t^2}(1-\frac{1}{2}\ns(p+q)t^2)  \Paren{1-\cos\Paren{2 \ns{(p-q) t}}} \tag{similar to Case 2} \notag\\
    &\geq \frac{1}{4\pi}\int_0^{\frac{\pi}{14\sqrt{\ns(p+q)}}} \frac{dt}{t^2} \Paren{1-\cos\Paren{2 \ns{(p-q) t}}} \notag\\
    &= \frac{\ns|p-q|}{4\pi}\int_0^{\frac{\pi \ns|p-q|}{14 \sqrt{\ns(p+q)}}} \frac{du}{u^2} \Paren{1-\cos 2u} \notag\\
    &\geq \frac{\pi \ns^2(p-q)^2}{14 \sqrt{\ns(p+q)}}\cdot \frac{1-\cos 2}{4\pi} \notag\\
    &\geq \frac{\ns^2(p-q)^2}{40\sqrt{\ns(p+q)}} \label{eq:bound:binomial:case3}
\end{align}
using, for the second-to-last inequality, that $\frac{\ns|p-q|}{\sqrt{\ns(p+q)}}<1 $ and that $t\mapsto \frac{1-\cos 2t}{t^2}$ is decreasing on $[0,1]$.
\end{description}
Combining~\cref{eq:bound:binomial:case1,eq:bound:binomial:case2,eq:bound:binomial:case3} concludes the proof.
\end{proof}

To conclude this section, we provide the proof of~\cref{claim:technical:inequality}:
\begin{proof}[{Proof of~\cref{claim:technical:inequality}}]
To give a lower bound on this term, we start from the following observation.
\begin{fact}
    \label{fact:modulus}
    For $p \in [0, \frac{1}{2}]$ and $t \in [0, \pi]$, the quantity $p \sin t$ is non-decreasing in $p$ and non-negative, and $r(t)$ is decreasing in $p$. Thus, we have $\abs{\frac{p\sin t}{r(t)} - \frac{q\sin t}{s(t)}} \geq \abs{\frac{p\sin t - q\sin t}{\max(r(t),s(t))}} > \abs{p\sin t -q\sin t} $, where the second inequality follows from $0 < r(t), s(t) < 1$.
\end{fact}

Now, since $\abs{\theta(t)-\eta(t)} = \abs{\arcsin\frac{p\sin(t)}{r(t)}-\arcsin\frac{q\sin(t)}{s(t)}}$ and the derivative of $\arcsin$ is always at least 1, ~\cref{fact:modulus} implies
\[
    2\ns\abs{\theta(t)-\eta(t)} 
    \geq 2\ns\abs{\frac{p\sin t}{r(t)} - \frac{q\sin t}{s(t)}} 
    > 2\ns\abs{p-q} \sin t 
\]
establishing~\eqref{eq:cos:lower_bound}.

To obtain an upper bound on $2\ns(\theta(t)-\eta(t))$, we will rely on the fact below.
\begin{fact}
     $0 \leq \frac{p\sin t}{r(t)}, \frac{q\sin t}{s(t)} \leq \frac{1}{2}$ when $0 \leq t \leq \pi$.
\end{fact}
\begin{proof}
    Since these two expressions are symmetric, we only need to prove $\frac{p\sin t}{r(t)} \leq \frac{1}{2}$ when $0 \leq t \leq \pi$. Because $p \sin t\leq \frac{1}{4}$ (as we assumed throughout $p\in[0,1/4]$), we only need to show that $r(t) \geq \frac{\sqrt{2}}{2}$. Note that $r(t)$ can be rewritten as $\sqrt{1-2p(1-p)(1-\cos t) }$, which is always at least $\sqrt{1-4p(1-p) }\geq \frac{1}{2}$.
\end{proof}
This will allow us to prove the upper bound,~\eqref{eq:cos:upper_bound}. Indeed, since the $\arcsin$ function is $\frac{2}{\sqrt{3}}$-Lipschitz on $[0,\frac{1}{2}]$ (as its derivative is $1/\sqrt{1-x^2}$), we have
\begin{align*}
    2\ns(\theta(t)-\eta(t)) &= 2\ns \abs{\arcsin\frac{p\sin(t)}{r(t)}-\arcsin \frac{q\sin(t)}{s(t)} } \notag \\
    & \leq \frac{4\ns}{\sqrt{3}} \abs{\frac{p\sin(t)}{r(t)}-\frac{q\sin(t)}{s(t)}} \notag \\
    &\leq \frac{8\ns}{\sqrt{3}} \abs{s(t) p \sin(t) - r(t) q \sin(t) } \tag{Since $r(t) s(t) \geq \frac{1}{2}$} \\
    &= \frac{8\ns}{\sqrt{3}}  \abs{\frac{s^2(t) p^2 - r^2(t) q^2 }{p s(t) + q r(t)}} \sin(t) \notag \\
    &\leq \frac{8\sqrt{2}\ns}{\sqrt{3}(p+q)} \abs{s^2(t) p^2 - r^2(t) q^2}  \sin(t) \tag{Since $r(t), s(t) \geq \frac{\sqrt{2}}{2}$} \\
    &= \frac{8\sqrt{2}\ns}{\sqrt{3}(p+q)} \abs{p^2 - q^2 - 2pq(p - q)(1-\cos t)}  \sin(t) \notag \\
    &\leq \frac{8\sqrt{2}\ns}{\sqrt{3}(p+q)} (\abs{p^2 - q^2} + 2pq\abs{p-q})  \sin(t) \notag \\
    &\leq \frac{16\sqrt{2}\ns}{\sqrt{3}} \abs{p - q}   \sin(t) \tag{$2pq \leq p^2+q^2 \leq p+q$}  \\
    &\leq \frac{16\sqrt{2}\ns}{\sqrt{3}} \abs{p - q} t \tag{By~\cref{fact:trigo}} 
\end{align*}
and $\frac{16\sqrt{2}}{\sqrt{3}} < 14$.
\end{proof}

\section{How to conclude: expectation gap and concentration of $Z$}
    \label{sec:getting:expectation:gap:from:there}
With~\cref{lemma:main:gap:binomial} in hand, we can establish~\cref{eq:exp:gap}, as in~\cite{DiakonikolasGKP21}. Assume $\totalvardist{\p}{\q} > \dst$, and 
define
\[
    S_1 = \mleft\{ i \in [\ab] : \min\Paren{ |\ns\p_i-\ns\q_i|, (\ns\p_i-\ns\q_i)^2, \frac{(\ns\p_i-\ns\q_i)^2}{\sqrt{\ns\p_i+\ns\q_i}} } = |\p_i-\q_i| \mright\}
\]
and similarly for $S_2, S_3$. By the lemma, we have
\begin{align*}
\bE{\p\q}{Z} &\geq \frac{1}{\ns} \sum_{i=1}^\ab \min\Paren{ \frac{1}{8}|\ns\p_i-\ns\q_i|, \frac{1}{16}(\ns\p_i-\ns\q_i)^2, \frac{1}{40}\frac{(\ns\p_i-\ns\q_i)^2}{\sqrt{\ns\p_i+\ns\q_i}} } \\
&= \frac{1}{8}\Paren{ \sum_{i\in S_1}  |\p_i-\q_i| + \frac{1}{2}\ns\sum_{i\in S_2} (\p_i-\q_i)^2
+ \frac{1}{5}\ns^{1/2}\sum_{i\in S_3} \frac{(\p_i-\q_i)^2}{\sqrt{\p_i+\q_i}} } \\
&\geq \frac{1}{8}\Paren{ \sum_{i\in S_1}  |\p_i-\q_i| + \frac{\ns}{2|S_2|}\Paren{\sum_{i\in S_2} |\p_i-\q_i|}^2
+ \frac{\ns^{1/2}}{5} \frac{\Paren{\sum_{i\in S_3} |\p_i-\q_i|}^2}{\sum_{i\in S_3}\sqrt{\p_i+\q_i}} } \\
&\geq \frac{1}{8}\Paren{ \sum_{i\in S_1}  |\p_i-\q_i| + \frac{\ns}{2\ab}\Paren{\sum_{i\in S_2} |\p_i-\q_i|}^2
+ \frac{1}{5}\Paren{\frac{\ns}{2\ab}}^{1/2}\Paren{\sum_{i\in S_3} |\p_i-\q_i|}^2 }
\,,
\end{align*}
where used the inequality
\begin{equation}
\sum_i \frac{a_i^2}{b_i} \geq \frac{\Paren{\sum_i |a_i|}^2}{\sum_i b_i}
\end{equation}
which holds for all $(a_i)_i$, and positive $(b_i)_i$; as well as $\sum_{i\in S_3} \sqrt{\p_i+\q_i} \leq \sqrt{|S_3| \sum_{i\in S_3} (\p_i+\q_i)} \leq \sqrt{2\ab}$ (by Jensen and concavity of $\sqrt{\cdot}$).

Since $\sum_{i=1}^\ab |\p_i-\q_i| > 2\dst$, at least one of $\sum_{i\in S_1} |\p_i-\q_i|,\sum_{i\in S_2} |\p_i-\q_i|,\sum_{i\in S_3} |\p_i-\q_i|$ must be at least $\frac{2}{3}\dst$, from which
\begin{align}
    \label{eq:expectation:gap:final}
    \bE{\p\q}{Z} &\geq \frac{1}{12}\min\Paren{\dst, \frac{\dst^2}{3}\frac{\ns}{\ab}, 
    \frac{\dst^2}{11}\sqrt{\frac{\ns}{\ab}} }
    \gtrsim
    \min\Paren{\dst, \dst^2\frac{\ns}{\ab}, 
    \dst^2\sqrt{\frac{\ns}{\ab}} }
\end{align}
as we wanted.

\paragraph{Concentration.} To prove concentration around the expectation, which is needed to obtain the (tight) sample complexity of closeness testing established in~\cite{DiakonikolasGKP21}, it then suffices to observe (looking at~\eqref{eq:statistic:z}) that changing any of the $4\ns$ samples can only change at most two of the $X_i,Y_i,X'_i,Y'_i$, and thus change (at most) two of the $\ab$ summands by $1/\ns$~--~i.e., change the value of $Z$ by at most $2/\ns$. This is the ``bounded difference property'' one needs to apply McDiarmid's inequality.

Thus, letting $\Delta^\ast$ denote the expectation gap established in~\eqref{eq:expectation:gap:final}, by McDiarmid the probability that $Z$ deviates from its expectation (in either the $\p=\q$ or $\totalvardist{\p}{\q}>\dst$ cases) by more than $\Delta^\ast/3$ is at most
\begin{equation}
    \bPr{ \abs{Z-\bEE{Z}} > \frac{1}{3}\Delta^\ast }
    \leq \exp\Paren{-\frac{2{\Delta^\ast}^2/9}{4\ns \cdot \Paren{2/\ns}^2}}
    = \exp\Paren{-\frac{\ns{\Delta^\ast}^2}{72}}\,.
\end{equation}
For this to be at most $\errprob$, one can verify based on the three regimes for the minimum defining $\Delta^\ast$ (cf.~\eqref{eq:expectation:gap:final}) that it suffices to have
\begin{equation}
    \ns \gtrsim \max\Paren{ \frac{\log(1/\errprob)}{\dst^2}, \frac{\ab^{2/3}\log^{1/3}(1/\errprob)}{\dst^{4/3}},
    \frac{\ab^{1/2}\log^{1/2}(1/\errprob)}{\dst^{2}}}
\end{equation}
which gives the (optimal) sample complexity for closeness testing as a function of all parameters (and where the hidden constants, albeit not optimized, are quite reasonable).

\section{Bonus: what about the Poissonized setting?}
\label{sec:poissonized}
We now show the generalizability of our approach, by establishing the analogue of~\cref{lemma:main:gap:binomial} for Poisson random variables, as considered (due to their proof technique) in~\cite{DiakonikolasGKP21}. Specifically, we show the following:
\begin{lemma}
    \label{lemma:main:gap:poisson}
Let $p,q \in [0,1/4]$, and $\ns \geq 16$. Suppose $X,X'\sim\poisson{\ns p}$ and $Y,Y'\sim\poisson{\ns q}$ are mutually independent. Then
\begin{align*}
    \bEE{|X-Y| + |X'-Y'| - |X-X'| - |Y-Y'|} 
    &\geq \min\Paren{ \frac{1}{20}(\mu-\lambda)^2, \frac{1}{5}|\mu-\lambda|, \frac{1}{7}\frac{(\mu-\lambda)^2}{\sqrt{\mu+\lambda}} }\,,
\end{align*}
where $\mu \eqdef \ns p = \bEE{X}$ and $\lambda\eqdef \ns q = \bEE{Y}$.
\end{lemma}
\begin{proof}
The proof is very similar to that of~\cref{lemma:main:gap:binomial}, and starts in an identical manner up to~\eqref{eq:cf}:
\begin{align*}
     \Delta &\eqdef \bEE{|X-Y| + |X'-Y'| - |X-X'| - |Y-Y'|} \\
    &= \frac{2}{\pi}\int_0^\infty \frac{dt}{t^2}\Paren{
    \bEE{e^{i tX}}\bEE{e^{-i tX}}
    +\bEE{e^{i tY}}\bEE{e^{-i tY}}
    -2\Re(\bEE{e^{i tX}}\bEE{e^{-i tY}})} 
\end{align*}
Using the expression of the CF of a Poisson random variable, along with the fact that
\[
    \Re\Paren{e^{\lambda(e^{i t}-1)}e^{\mu(e^{-i t}-1)}}
    = \Re\Paren{ e^{(\lambda+\mu)(\cos t - 1) + i(\lambda-\mu) \sin t} }
    = e^{(\lambda+\mu)(\cos t - 1)\cdot} \Re\Paren{ e^{i(\lambda-\mu) \sin t} }
\]
we then get
\begin{align}
     \Delta 
&= \frac{2}{\pi}\int_0^\infty \frac{dt}{t^2}\Paren{
    e^{2\lambda(\cos t-1)}
    + e^{2\mu (\cos t-1)}
    -2e^{(\lambda+\mu)(\cos t - 1)}\cos((\lambda-\mu)\sin t) } \notag\\
    &\geq \frac{2}{\pi}\int_0^\infty \frac{dt}{t^2}\Paren{
    2e^{(\lambda+\mu)(\cos t-1)}
    -2e^{(\lambda+\mu)(\cos t - 1)}\cos((\lambda-\mu)\sin t) } \tag{AM--GM}\\
&= \frac{4}{\pi}\int_0^\infty \frac{dt}{t^2} e^{(\lambda+\mu)(\cos t-1)} \Paren{
    1 - \cos((\lambda-\mu)\sin t) } \label{eq:post:Jensen:equality}
\end{align}
Now, as in~\cref{sec:multinomial}, we have three cases to consider.
\begin{itemize}
    \item Suppose $\lambda+\mu \leq 1$. Then, since $(\lambda+\mu)(\cos t - 1) \geq -2$ for all $t$,  by~\eqref{eq:post:Jensen:equality},
\begin{align*}
     \Delta
&\geq \frac{4}{\pi e^2}\int_0^{\frac{\pi}{2}} \frac{dt}{t^2} \Paren{
    1 - \cos((\lambda-\mu)\sin t) }
\geq \frac{4}{\pi e^2}\int_0^{\frac{\pi}{2}} \frac{dt}{t^2} \Paren{
    1 - \cos\Paren{\frac{2}{\pi}(\lambda-\mu) t} }
\end{align*}
the second inequality by~\cref{fact:trigo}. By a change of variable, we get
\begin{align}
     \Delta
&\geq |\lambda-\mu|\cdot \frac{8}{\pi^2 e^2}\int_0^{|\lambda-\mu|} \frac{du(1 - \cos u)}{u^2} 
\geq (\lambda-\mu)^2\cdot \frac{8}{\pi^2 e^2}\cdot (1-\cos 1) \geq \frac{1}{20}(\lambda-\mu)^2
\end{align}
using that $|\lambda-\mu|\leq \lambda+\mu\leq 1$ and monotonicity of $u\mapsto \frac{1-\cos u}{u^2}$ on $[0,1]$ to bound the integrand.
\item Suppose that $\lambda+\mu>1$ and $|\lambda-\mu| \geq \sqrt{\lambda+\mu}$ (and so $|\lambda-\mu| > 1$). Then we can bound~\eqref{eq:post:Jensen:equality} as
\begin{align*}
     \Delta
&\geq \frac{4}{\pi}\int_0^\infty \frac{dt}{t^2} e^{-(\lambda-\mu)^2(1-\cos t)} \Paren{1 - \cos((\lambda-\mu)\sin t) }\\
&\geq \frac{4}{\pi}\int_0^{\frac{1}{|\lambda-\mu|}} \frac{dt}{t^2} e^{-(\lambda-\mu)^2\cdot \frac{t^2}{2}} \Paren{1 - \cos\Paren{\frac{2}{\pi}(\lambda-\mu) t} }
\end{align*}
using again $\sin t \geq \frac{2}{\pi}t$ for $t \leq \frac{\pi}{2}$ (note that this is a fortiori true for $t \leq 1/|\lambda-\mu| < 1$), and $1-\cos t \leq \frac{t^2}{2}$. By another change of variable $u=|\lambda-\mu|t$, we get
\begin{align}
     \Delta
&\geq \frac{4}{\pi} |\lambda-\mu|\int_0^1 \frac{du}{u^2} e^{-\frac{1}{2}u^2} \Paren{1 - \cos\Paren{\frac{2}{\pi}u} } \geq \frac{1}{5}|\lambda-\mu|.
\end{align}

\item Finally, suppose that $\lambda+\mu>1$ and $|\lambda-\mu| < \sqrt{\lambda+\mu}$. Then, since for $t\leq \frac{1}{\sqrt{\lambda+\mu}} < 1$ we again have $\sin t \geq \frac{2}{\pi}t$ and $|\lambda-\mu|t < 1$. From~\eqref{eq:post:Jensen:equality} we can write
\begin{align*}
    \Delta
&\geq \frac{4}{\pi}\int_0^{\frac{1}{\sqrt{\lambda+\mu}}} \frac{dt}{t^2} e^{-(\lambda+\mu)(1-\cos t)} \Paren{1- \cos((\lambda-\mu)\sin t) } \\
&\geq \frac{4}{\pi}\int_0^{\frac{1}{\sqrt{\lambda+\mu}}} \frac{dt}{t^2} e^{-(\lambda+\mu)\frac{t^2}{2}} \Paren{1- \cos\Paren{\frac{2}{\pi}(\lambda-\mu) t} } \\
&\geq \frac{4}{\pi} e^{-\frac{1}{2}}  \int_0^{\frac{1}{\sqrt{\lambda+\mu}}} \frac{dt}{t^2} \Paren{1- \cos\Paren{\frac{2}{\pi}(\lambda-\mu) t} } \\
&= |\lambda-\mu|\cdot  \frac{8}{\pi^2} e^{-\frac{1}{2}}  \int_0^{\frac{2}{\pi}\frac{|\lambda-\mu|}{\sqrt{\lambda+\mu}}} \frac{du\Paren{1- \cos u }}{u^2}  \\
&\geq |\lambda-\mu|\cdot  \frac{4}{\pi} e^{-\frac{1}{2}}  \cdot \frac{|\lambda-\mu|}{\sqrt{\lambda+\mu}} \cdot \Paren{1- \cos \frac{2}{\pi} }\,,
\end{align*}
using, for the last inequality, that $\frac{|\lambda-\mu|}{\sqrt{\lambda+\mu}}<1 $ and that $\frac{1-\cos t}{t^2}$ is decreasing on $[0,1]$.  Since $\frac{4}{\pi e^{1/2}}\Paren{1- \cos \frac{2}{\pi} } > 1/7$, we get that in this case
\begin{equation}
    \Delta \geq \frac{1}{7} \cdot \frac{(\lambda-\mu)^2}{\sqrt{\lambda+\mu}}
\end{equation}
\end{itemize}
This concludes the proof.
\end{proof}
\printbibliography
\end{document}